\documentclass[a4paper,twoside]{article}

\usepackage{amsmath}
\usepackage{amssymb}
\usepackage{url}
\usepackage{enumerate}
\usepackage{color}
\usepackage{graphicx}
\usepackage{epsfig}
\usepackage{subfigure}
\usepackage{calc}
\usepackage{amstext}
\usepackage{amsthm}
\usepackage{multicol}
\usepackage{pslatex}
\usepackage{apalike}
\usepackage{SCITEPRESS}     

\newcommand{\At}{\mathcal At}
\newcommand{\body}[1]{\mathsf{body}\left(#1\right)}
\newcommand{\head}[1]{\mathsf{head}\left(#1\right)}
\renewcommand{\I}{\mathcal I}
\newcommand{\Ieta}{\langle\mathcal I,\eta\rangle}
\newcommand{\J}{\mathcal J}

\newcommand{\U}{\mathcal U}
\newcommand{\neff}[1]{\mathsf{ne}\left(\I,\I\circ#1\right)}
\newcommand{\neffU}{\neff{\U}}
\newcommand{\lit}{\mathsf{lit}}
\newcommand{\nup}{\mathsf{nup}}

\newcommand{\ua}{\mathsf{ua}}
\newcommand{\pnot}{\mathsf{not}\ }
\newcommand{\tool}{\texttt{repAIrC}}

\newcommand{\class}[1]{\texttt{#1}}

\newtheorem{theorem}{Theorem}
\newtheorem{example}{Example}
\newtheorem{lemma}{Lemma}
\newtheorem{definition}{Definition}

\begin{document}

\title{\tool: A Tool for Ensuring Data Consistency\subtitle{by Means of Active Integrity Constraints}}
\author{
  \authorname{Lu\'\i s Cruz-Filipe\sup{1},
    Michael Franz\sup{1},
    Artavazd Hakhverdyan\sup{1},
    Marta Ludovico\sup{2},
    Isabel Nunes\sup{2}
    and Peter Schneider-Kamp\sup{1}}
  \affiliation{\sup{1}Dept.\ of Mathematics and Computer Science, University of Southern Denmark, Campusvej 55, 5230 Odense M, Denmark}
  \affiliation{\sup{2}Faculdade de Ci\^encias da Universidade de Lisboa, Campo Grande, 1749-016 Lisboa, Portugal}
  \email{lcf@imada.sdu.dk, mf@bfdata.dk, artavazd19@gmail.com, marta.al.ludovico@gmail.com, in@fc.ul.pt, petersk@imada.sdu.dk}
}


\abstract{Consistency of knowledge repositories is of prime importance in
  organization management.
  Integrity constraints are a well-known vehicle for specifying data consistency
  requirements in knowledge bases; in particular, active
  integrity constraints go one step further, allowing the specification of
  preferred ways to overcome inconsistent situations in the context of database
  management.\\
  This paper describes a tool to validate an SQL database with respect to a
  given set of active integrity constraints, proposing possible repairs in case
  the database is inconsistent.
  The tool is able to work with the different kinds of repairs proposed in the
  literature, namely simple, founded, well-founded and justified repairs.
  It also implements strategies for parallelizing the search for them, allowing
  the user both to compute partitions of independent or stratified active
  integrity constraints, and to apply these partitions to find repairs of
  inconsistent databases efficiently in parallel.}

\onecolumn \maketitle \normalsize \vfill

\section{\uppercase{Introduction}}

\noindent
There is a generalized consensus that knowledge repositories are a key
ingredient in the whole process of Knowledge Management,
cf.~\cite{Duhon98,Konig2012}.
Furthermore, being able to rely upon the consistency of the information they
provide is paramount to any business whatsoever.
Databases and database management systems, by far the most common framework for
knowledge storage and retrieval, have been around for many years now,
and have evolved substantially, at pace with information technology.
In this paper, we are focusing on the important aspect of database consistency.

Typical database management systems allow the user to specify
integrity constraints on the data as logical statements that are
required to be satisfied at any given point in time.
The classical problem is how to guarantee that such constraints still
hold after updating databases~\cite{Abiteboul1988}, and what repairs
have to be made when the constraints are violated~\cite{Katsuno1991},
without making any assumptions about how the inconsistencies came about.
Repairing an inconsistent database~\cite{Eiter1992} is a highly
complex process; also, it is widely accepted that human
intervention is often necessary to choose an adequate repair.
That said, every progress towards automation in this field is nevertheless important.

In particular, the framework of active integrity
constraints~\cite{Flesca2004,Caroprese2011} was introduced more recently with the
goal of giving operational mechanisms to compute repairs of
inconsistent databases.
This framework has subsequently been extended to consider
preferences~\cite{Caroprese2007} and to find ``best'' repairs
automatically~\cite{CEGN2013} and efficiently~\cite{lcf:14}.

Active integrity constraints (AICs) seem to be a promising framework
for the purpose of achieving reliability in information retrieval:
\begin{itemize}
\item AICs are expressive enough to encompass the majority of
  integrity constraints that are typically found in practice;
\item AICs allow the definition of preferred ways to calculate
  repairs, through specific actions to be taken in specific
  inconsistent situations;
\item AICs provide mechanisms to resolve inconsistencies while the
  database is in use;
\item AICs can enhance databases to provide a basis for self-healing
  autonomic systems.
\end{itemize}
To the best of our knowledge, no real-world implementation of an
AIC--enhanced database system exists today.
This paper presents a prototype tool that implements the tree--based
algorithms for computing repairs presented
in~\cite{Caroprese2011,CEGN2013}.
While not yet ready for productive deployment, this implementation can
work successfully with database management systems working in the SQL
framework, and is readily extendible to other (nearly arbitrary)
database management systems thanks to its modular design.

This paper is structured as follows.
Section~\ref{sec:background} recapitulates previous work on active
integrity constraints and repair trees.
Section~\ref{sec:tool} introduces our tool, \tool, and describes its
implementation, focusing on the new theoretical results that were
necessary to bridge the gap between theory and practice.
Section~\ref{sec:parallel} then discusses how parallel computation
capabilities are incorporated in \tool\ to make the search for repairs
more efficient.
Section~\ref{sec:concl} summarizes our achievements and gives a brief
outlook into future developments.

\section{\uppercase{Active integrity constraints}}
\label{sec:background}

\noindent
Active integrity constraints (AICs) were introduced in~\cite{Flesca2004} and
further explored in~\cite{Caroprese2009,Caroprese2011}, which define the basic
concepts and prove complexity bounds for the problem of repairing inconsistent
databases.
These authors introduce declarative semantics for different types of repairs,
obtaining their complexity results by means of a translation into revision
programming.
In practice, however, this does not yield algorithms that are applicable to
real-life databases; for this reason, a direct operational semantics
for AICs was proposed in~\cite{CEGN2013}, presenting database-oriented
algorithms for finding repairs.
The present paper describes a tool that can actually execute these algorithms in
collaboration with an SQL database management system.

\subsection{Syntax and Declarative Semantics}

For the purpose of this work, we can view a database simply as a set of atomic
formulas over a typed function-free
first-order signature $\Sigma$, which we will assume throughout to be fixed.
Let $\At$ be the set of closed atomic formulas over $\Sigma$.
A database $\I$ \emph{entails} literal $L$, $\I\models L$, if $L\in\At$ and
$L\in\I$, or if $L$ is $\pnot a$ with $a\in\At$ and $a\notin\I$.

An integrity constraint is a clause $$L_1,\ldots,L_m \supset \bot$$ where each
$L_i$ is a literal over $\Sigma$, with intended semantics that
$\forall(L_1\wedge\ldots\wedge L_m)$ should not hold.
As is usual in logic programming, we require that if $L_i$ contains a negated variable $x$, then $x$
already occurs in $L_1,\ldots,L_{i-1}$.
We say that $\I$ \emph{satisfies} integrity constraint $r$, $\I\models r$, if,
for every instantiation $\theta$ of the variables in $r$, it is the case that
$\I\not\models L\theta$ for some $L$ in $r$; and $\I$ satisfies a set $\eta$ of
integrity constraints, $\I\models\eta$, if it satisfies each integrity
constraint in $\eta$.

If $\I\not\models\eta$, then $\I$ may be updated through \emph{update actions}
of the form $+a$ and $-a$, where $a\in\At$, stating that $a$ is to be inserted
in or deleted from $\I$, respectively.
A set of update actions $\U$ is \emph{consistent} if it does not contain both
$+a$ and $-a$, for any $a\in\At$; in this case, $\I$ can be updated by $\U$,
yielding the database
$$\I\circ\U = (\I\cup\left\{a\mid{+a}\in\U\right\})\setminus \left\{a\mid{-a}\in\U\right\}\,.$$
The problem of database repair is to find $\U$ such that $\I\circ\U\models\eta$.

\begin{definition}
  Let $\I$ be a database and $\eta$ a set of integrity constraints.
  A \emph{weak repair} for $\Ieta$ is a consistent set $\U$ of update actions
  such that:
  (i)~every action in $\U$ changes $\I$; and
  (ii)~$\I\circ\U\models\eta$.
  A \emph{repair} for $\Ieta$ is a weak repair $\U$ for $\Ieta$ that is minimal w.r.t.\ set
  inclusion.
\end{definition}
The distinction between weak repairs and repairs embodies the
standard principle of \emph{minimality of change}~\cite{Winslett1990}.

The problem of deciding whether there exists a (weak) repair for an
inconsistent database is
$NP$-complete~\cite{Caroprese2011}.
Furthermore, simply detecting that a
database is inconsistent does not give any information on how it can be
repaired.
In order to address this issue, those authors proposed active integrity constraints (AICs),
which guide the process of selection of a repair by pairing literals with the
corresponding update actions.

In the syntax of AICs, we extend the notion of update action by allowing
variables.
Given an action $\alpha$, the literal corresponding to it is $\lit(\alpha)$,
defined as $a$ if $\alpha={+a}$ and $\pnot a$ if $\alpha={-a}$; conversely, the
update action corresponding to a literal $L$, $\ua(L)$, is $+a$ if $L=a$ and
$-a$ if $L=\pnot a$.
The \emph{dual} of $a$ is $\pnot a$, and conversely; the dual of $L$ is denoted
$L^D$.
An \emph{active integrity constraint} is thus an expression $r$ of the form
$$L_1,\ldots,L_m \supset \alpha_1\mid\ldots\mid\alpha_k$$
where the $L_i$ (in the \emph{body} of $r$, $\body r$) are literals and the
$\alpha_j$ (in the \emph{head} of $r$, $\head r$) are update actions, such that
\[\left\{\lit(\alpha_1)^D,\ldots,\lit(\alpha_k)^D\right\}\subseteq\left\{L_1,\ldots,L_m\right\}\,.\]
The set $\lit(\head r)^D$ contains the \emph{updatable} literals of $r$.
The \emph{non-updatable} literals of $r$ form the set
$\nup(r)=\body r\setminus\lit\left(\head r\right)^D$.

The natural semantics for AICs restricts the notion of weak repair.

\begin{definition}
  Let $\I$ be a database, $\eta$ a set of AICs and $\U$ be a (weak) repair for
  $\Ieta$.
  Then $\U$ is a \emph{founded (weak) repair} for $\Ieta$ if, for every action
  $\alpha\in\U$, there is a closed instance $r'$ of $r\in\eta$ such that
  $\alpha\in\head{r'}$ and $\I\circ\U\models L$ for every
  $L\in\body{r'}\setminus\left\{\lit(\alpha)^D\right\}$.
\end{definition}

The problem of deciding whether there exists a weak founded repair for an
inconsistent database is again $NP$-complete, while the similar problem for
founded repairs is $\Sigma^P_2$-complete.
Despite their natural definition, founded repairs can include circular support
for actions, which can be undesirable; this led to the introduction of justified
repairs~\cite{Caroprese2011}.

We say that a set $\U$ of update actions is \emph{closed} under $r$ if
$\nup(r)\subseteq\lit(\U)$ implies $\head r\cap\U\neq\emptyset$, and it is
closed under a set $\eta$ of AICs if it is closed under every closed instance of
every rule in $\eta$.
In particular, every founded weak repair for $\Ieta$ is by definition
closed under~$\eta$.

A closed update action $+a$ (resp.\ $-a$) is a \emph{no-effect} action w.r.t.\
$(\I,\I\circ\U)$ if $a\in\I\cap(\I\circ\U)$
(resp.\ $a\notin\I\cup(\I\circ\U)$).
The set of all no-effect actions w.r.t.\ $(\I,\I\circ\U)$ is denoted by
$\neffU$.
A set of update actions $\U$ is a justified action set if it coincides with the
set of update actions forced by the set of AICs and the database before and
after applying $\U$~\cite{Caroprese2011}.

\begin{definition}
  Let $\I$ be a database and $\eta$ a set of AICs.
  A consistent set $\U$ of update actions is a \emph{justified action set} for
  $\Ieta$ if it is a minimal set of update actions containing $\neffU$ and
  closed under $\eta$.
  If $\U$ is a justified action set for $\Ieta$, then $\U\setminus\neffU$ is a
  justified weak repair for $\Ieta$.
\end{definition}
In particular, it has been shown that justified repairs are always
founded~\cite{Caroprese2011}.
The problem of deciding whether there exist justified weak repairs or justified repairs
for $\Ieta$ is again a
$\Sigma^P_2$-complete problem, becoming $NP$-complete if one restricts the AICs
to contain only one action in their head (\emph{normal} AICs).

\subsection{Operational Semantics}

The declarative semantics of AICs is not very satisfactory, as it does not
capture the operational nature of rules.
In particular, the quantification over all no-effect actions in the definition
of justified action set poses a practical problem.
Therefore, an operational semantics for AICs was proposed in~\cite{CEGN2013},
which we now summarize.

\begin{definition}
  Let $\I$ be a database and $\eta$ be a set of AICs.
  \begin{itemize}
  \item The \emph{repair tree} for $\Ieta$, $T_{\Ieta}$, is a labeled tree
    where: nodes are sets of update actions; each edge is labeled with a closed
    instance of a rule in $\eta$;
    the root is $\emptyset$; and for each consistent node $n$ and closed
    instance $r$ of a rule in $\eta$, if $\I\circ n\not\models r$ then for each
    $L\in\body{r}$ the set $n'=n\cup\left\{\ua(L)^D\right\}$ is a child of $n$,
    with the edge from $n$ to $n'$ labeled by $r$.
  \item The \emph{founded repair tree} for $\Ieta$, $T^f_{\Ieta}$, is
    constructed as $T_{\Ieta}$ but requiring that $\ua(L)$ occur in the head of
    some closed instance of a rule in $\eta$.
  \item The \emph{well-founded repair tree} for $\Ieta$, $T^{wf}_{\Ieta}$, is
    also constructed as $T_{\Ieta}$ but requiring that $\ua(L)$ occur in the
    head of the rule being applied.
  \item The \emph{justified repair tree} for $\Ieta$, $T^j_{\Ieta}$, has nodes
    that are \emph{pairs} of sets of update actions $\langle\U,\J\rangle$, with
    root $\langle\emptyset,\emptyset\rangle$.
    For each node $n$ and closed instance $r$ of a rule in $\eta$, if
    $\I\circ\U_n\not\models r$, then for each $\alpha\in\head{r}$ there is a
    descendant $n'$ of $n$, with the edge from $n$ to $n'$ labeled by $r$,
    where:
    $\U_{n'} = \U_n\cup\left\{\alpha\right\}$;
    and $\J_{n'} = \left(\J_n\cup\{\ua(\nup(r))\}\right)\setminus\U_n$.
  \end{itemize}
\end{definition}

The properties of repair trees are summarized in the following results, proved
in~\cite{CEGN2013}.
\begin{theorem}
  \label{thm:trees}
  Let $\I$ be a database and $\eta$ be a set of AICs.  Then:
  \begin{enumerate}
  \item $T_{\Ieta}$ is finite.
  \item Every consistent leaf of $T_{\Ieta}$ is labeled by a weak repair for $\Ieta$.
  \item If $\U$ is a repair for $\Ieta$, then there is a branch of $T_{\Ieta}$
    ending with a leaf labeled by $\U$.
  \item If $\U$ is a founded repair for $\Ieta$, then there
    is a branch of $T^f_{\Ieta}$ ending with a leaf
    labeled by $\U$.
  \item If $\U$ is a justified repair for $\Ieta$, then there
    is a branch of $T^j_{\Ieta}$ ending with a leaf
    labeled by $\U$.
  \item If $\eta$ is a set of normal AICs and $\langle\U,\J\rangle$ is a leaf of
    $T^j_{\Ieta}$ with $\U$ consistent and $\U\cap\J=\emptyset$, then $\U$ is a
    justified repair for $\Ieta$.
  \end{enumerate}
\end{theorem}
Not all leaves will correspond to repairs of the desired kind; in particular,
there may be weak repairs in repair trees.
Also, both $T^f_{\Ieta}$ and $T^j_{\Ieta}$ typically contain leaves that do not
correspond to founded or justified (weak) repairs -- otherwise the problem of
deciding whether there exists a founded or justified weak repair for $\Ieta$
would be solvable in non-deterministic polynomial time.
The leaves of the well-founded repair tree for $\Ieta$ correspond to a new type of weak repairs,
called \emph{well-founded weak repairs}, not
considered in the original works on AICs.

\subsection{Parallel Computation of Repairs}
\label{ssec:par}

The computation of founded or justified repairs can be improved by dividing the
set of AICs into independent sets that can be processed
independently, simply merging the computed repairs at the end~\cite{lcf:14}.
Here, we adapt the definitions given therein to the first-order scenario.
Two sets of AICs $\eta_1$ and $\eta_2$ are independent if the same atom does not
occur in a literal in the body of a closed instance of two distinct rules
$r_1\in\eta_1$ and $r_2\in\eta_2$.
If $\eta_1$ and $\eta_2$ are independent, then repairs for
$\langle I,\eta_1\cup\eta_2\rangle$ are exactly the unions of a repair for
$\langle\I,\eta_1\rangle$ and $\langle\I,\eta_2\rangle$; furthermore,
the result still holds if one considers founded, well-founded or
justified repairs.

If an atom occurs in a literal in the body of a closed instance of a rule in
$\eta_2$ and in an action in the head of a closed instance of a rule in
$\eta_1$, but not conversely, then we say that $\eta_1$ \emph{precedes}
$\eta_2$.
Founded/justified (but not well-founded) repairs for $\eta_1\cup\eta_2$ can be
computed in a stratified way, by first repairing $\I$ w.r.t.~$\eta_1$, and then
repairing the result w.r.t.~$\eta_2$.

Splitting a set of AICs into independent sets or stratifying it can be solved
using standard algorithms on graphs, as we describe in
Section~\ref{sec:parallel}.

\section{\uppercase{The tool}}
\label{sec:tool}

\noindent
The tool \tool\ is implemented in Java, and its simplified UML class diagram can
be seen in Figure~\ref{fig:classes}.
Structurally, this tool can be split into four main separate components,
centered on the four classes marked in bold in that figure.
\begin{itemize}
\item Objects of type \class{AIC} implement active integrity constraints.
\item Implementations of interface \class{DB} provide the necessary tools to
  interact with a particular database management system; currently, we provide
  functionality for SQL databases supported by JDBC.
\item Objects of type \class{RepairTree} correspond to concrete repair trees;
  their exact type will be the subclass corresponding to a particular kind of
  repairs.
\item Class \class{RunRepairGUI} provides the graphical interface to interact
  with the user.
\end{itemize}

\begin{figure*}[!ht]
  \centering
  \scriptsize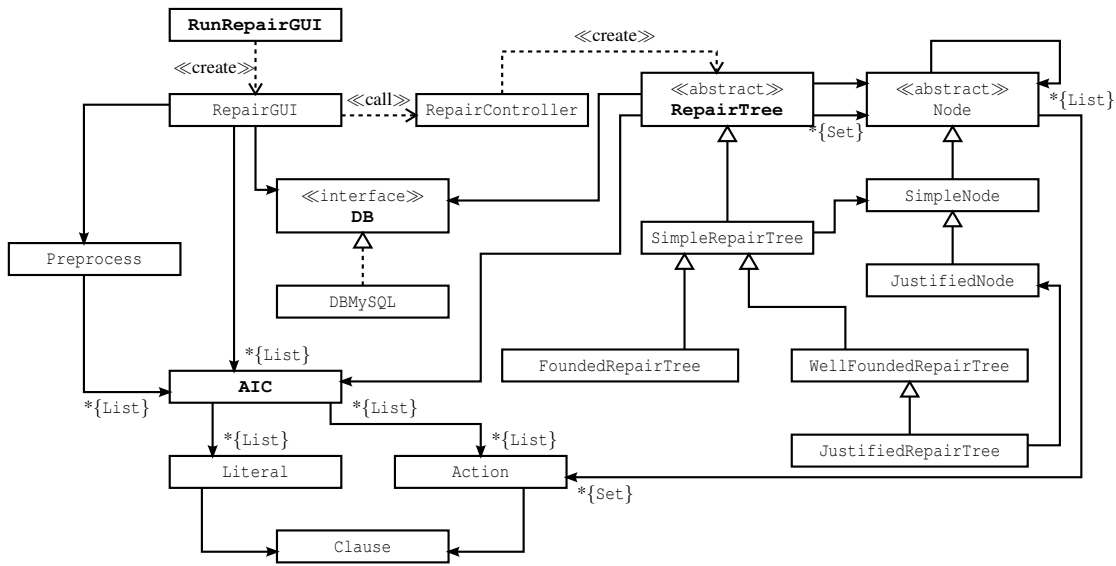
  \caption{Class diagram for \tool.}
  \label{fig:classes}
\end{figure*}

An important design aspect has to do with extensibility and modularity.
A first prototype focused on the construction of repair trees, and used simple
text files to mimick databases as lists of propositional atoms, in the style
of~\cite{Caroprese2011,CEGN2013}.
Later, parallelization capabilities were added (as explained in
Section~\ref{sec:parallel}), requiring changes only to \class{RepairController}
-- the class that controls the execution of the whole process.
Likewise, the extension of \tool\ to SQL databases and the addition of the
stratification mechanism only required localized changes in the classes directly
concerned with those processes.

The next subsections detail the implementation of the classes \class{AIC}, \class{DB},
\class{RepairTree} and \class{RunRepairTreeGUI}.

\subsection{Representing Active Integrity Constraints}

\noindent
In the practical setting, it makes sense to diverge a little from the
theoretical definition of AICs.
\begin{itemize}
\item Real-world tables found in DBs contain many columns, most of which are typically irrelevant for a given integrity constraint.
\item The columns of a table are not static, i.e., columns are usually added or removed during a database's lifecycle.
\item The order of columns in a table should not matter, as they are identified by a unique column name.
\end{itemize}
To deal pragmatically with these three aspects, we will write atoms using a more database-oriented notation, allowing
the arguments to be provided in any order, but requiring that the column names
be provided.
The special token \verb+$+ is used as first character of a variable.
So, for example, the literal \verb+hasInsurance(firstName=$X, type='basic')+
will match any entry in table \verb+hasInsurance+ having value \verb+basic+ in
column \verb+type+ and any value in column \verb+firstName+; this table may
additionally have other columns.
Negative literals are preceded by the keyword \verb+NOT+, while actions must
begin with \verb-+- or \verb+-+.
Literals and actions are separated by commas, and the body and head of an AIC
are separated by \verb+->+.
The AIC is finished when \verb+;+ is encountered, thus allowing constraints to
span several lines.

AICs are provided in a text file, which is parsed by a parser generated
automatically using JavaCC and transformed into objects of type \class{AIC}.
These contain a body and a head, which are respectively
\class{List<Literal>} and \class{List<Action>};
for consistency with the underlying theory, \class{Literal} and \class{Action}
are implemented separately, although their objects are isomorphic:
they contain an object of type \class{Clause} (which consists of the
name of a table in the database and a list of pairs column name/value)
and a flag indicating whether they are positive/negated (literals) or additions/removals
(actions).

\begin{example}
\label{ex:AIC1}
Consider the following active integrity constraints for an employee database.
The first states
that the boss (as specified in the \textnormal{\sf category} table) cannot be a junior employee (i.e., have an entry in the \textnormal{\sf junior} table);
the second states that every junior employee must have some basic insurance (as specified in the \textnormal{\sf insured} table).
\[
  \mathsf{junior}(X), \mathsf{category}(\mathsf{boss},X) \supset
  -\mathsf{junior}(X)
\]
\begin{multline*}
  \label{AICex2}
  \mathsf{junior}(X), \mathsf{not\ insured}(X,\mathsf{basic})
  \\ \supset +\mathsf{insured}(X,\mathsf{basic})
\end{multline*}

These are written in the concrete text-based syntax of the \tool\ tool as
{\small
\begin{verbatim}
junior(id = $X),
  category(type = boss, empId = $X)
  -> - junior(id = $X);

junior(id = $X),
  NOT insured(empId = $X, type = basic)
  -> + insured(empId = $X, type = basic);
\end{verbatim}}
\noindent respectively, assuming the corresponding column names for
the atributes.
Note that, thanks to our usage of explicit column naming, the column names for the same variable need not have
identical designations.
\end{example}

\subsection{Interfacing with the Database}
\label{ssec:interface}

Database operations (queries and updates) are defined in the
\class{DB} interface, which contains the following methods.
\begin{itemize}
\item \verb+getUpdateActions(AIC aic)+: queries the database for all the instances
  of \verb+aic+ that are not satisfied in its current state, returning a
  \verb+Collection<Collection<Action>>+ that contains the corresponding
  instantiations of the head of \verb+aic+.
\item \verb+update(Collection<Action> actions)+: applies all update actions in
  \verb+actions+ to the database (void).
\item \verb+undo(Collection<Action> actions)+: undoes the effect of all update
  actions in \verb+actions+ (void).
\item \verb+aicsCompatible(Collection<AIC> aics)+: checks that all the elements of
  \verb+aics+ are compatible with the structure of the database.
\item \verb+disconnect()+: disconnects from the database (void). The
  connection is established when the object is originally constructed.
\end{itemize}

Some of these methods require more detailed comments.
The construction of the repair tree also requires that the database be changed
interactively, but upon conclusion the database should be returned to its
original state.
In theory, this would be achievable by applying the \verb+update+ method with
the duals of the actions that were used to change the database; but this turns
out not to be the case for deletion actions.
Since the AICs may underspecify the entries in the database (because some fields
are left implicit), the implementation of \verb+update+ must take care to
store the values of all rows that are deleted from the database.
In turn, the \verb+undo+ method will read this information every time it has to
undo a deletion action, in order to find out exactly what entries to re-add.

The method \verb+aicsCompatible+ is necessary because the AICs are given
independently of the database, but they must be compatible with its structure --
otherwise, all queries will return errors.
Including this method in the interface allows the AICs to be tested before any queries are made,
thus significantly reducing the number of exceptions that can occur during
program execution.





  


Currently, \tool\ includes an implementation \verb+DBMySQL+ of \verb+DB+, which
works with SQL databases.
The interaction between \tool\ and the database is achieved by means of JDBC, a
Java database connectivity technology able to interface with nearly all existing
SQL databases.
In order to determine whether an AIC is satisfied by a database, method
\verb+getUpdateActions+ first builds a single SQL query corresponding to the body of
the AIC.
This method builds two separate \verb+SELECT+ statements, one for the positive
and another for the negative literals in the body of the AIC.
Each time a new variable is found, the table and column where it occurs are
stored, so that future references to the same variable in a positive literal can
be unified by using inner joins.
The \verb+select+ statement for the negative literals is then connected to the
other one using a \verb+WHERE NOT EXISTS+ condition.
Variables in the negative literals must necessarily appear first in a positive
literal in the same AIC; therefore, they can then be connected by a \verb+WHERE+
clause instead of an inner join.

\begin{example}
The bodies of the integrity constraints in Example~\ref{ex:AIC1} generate the
following SQL queries.
{\small
\begin{verbatim}
SELECT * FROM junior
  INNER JOIN dept_emp
  ON junior.id=category.empId
  WHERE category.type=`boss'

SELECT * FROM junior
  WHERE NOT EXISTS
 (SELECT * FROM insured
    WHERE insured.empId=junior.id
    AND insured.type=`basic')
\end{verbatim}}
\end{example}

\subsection{Implementing Repair Trees}

The implementation of the repair trees directly follows the algorithms described
in Section~\ref{sec:background}.
Different types of repair trees are implemented using inheritance, so that most
of the code can be reused in the more complex trees.
The trees are constructed in a breadth-first manner, and all non-contradictory
leaves that are found are stored in a list.
At the end, this list is pruned so that only the minimal elements (w.r.t.~set
inclusion) remain -- as these are the ones that correspond to repairs.

While constructing the tree, the database has to be temporarily updated and
restored.
Indeed, to calculate the descendants of a node, we first need to evaluate all
AICs at that node in order to determine which ones are violated; this requires
querying a modified version of the database that takes into account the update
actions in the current node.

In order to avoid concurrency issues, these updates are performed in a transaction-style way,
where we update the database, perform the necessary SQL queries, and rollback to
the original state, guaranteeing that other threads interacting with the
database during this process neither see the modifications nor lead to inconsistent repair trees.
This becomes of particular interest when the parallel processing tools described
in Section~\ref{sec:parallel} are put into place.
Although this adds some overhead to the execution time,
at the end of that section we discuss why scalability is not a practically relevant concern.

After finding all the leaves of the repair tree, a further step is needed in the
case one is looking for founded or justified repairs, as the corresponding trees
may contain leaves that do not correspond to repairs with the desired property.
This step is skipped if all AICs are normal, in view of the results
from~\cite{CEGN2013}.
For founded repairs, we directly apply the definition: for each action $\alpha$,
check that there is an AIC with $\alpha$ in its head and such that all other
literals in its body are satisfied by the database.

For justified repairs, the validation step is less obvious.
Directly following the definition requires constructing the set of
no-effect actions, which is essentially as large as the database, and iterating
over subsets of this set.
This is obviously not possible to do in practical settings.
Therefore, we use some criteria to simplify this step.

\begin{lemma}
  If a rule $r$ was not applied in the branch leading to $\U$, then $\U$ is
  closed under $r$.
\end{lemma}
\begin{proof}
  Suppose that $r$ was never applied and assume $\nup(r)\subseteq\neff\U$.
  Then necessarily $\head r\cap\neff\U\neq\emptyset$, otherwise $r$ would be
  applicable and $\U$ would not be a repair.
\end{proof}

By construction, $\U$ is also closed for all rules applied in the branch leading
to it.

Let $\U$ be a candidate justified weak repair.
In order to test it, we need to show that $\U\cup\neff\U$ is a justified
action set (see~\cite{CEGN2013}), which requires iterating over all subsets of
$\U\cup\neff\U$ that contain $\neff\U$.
Clearly this can be achieved by iterating over subsets of $\U$.

But if $\U^\ast\subseteq \U$, then $\nup(r)\cap \U^\ast=\emptyset$; this
allows us to simplify the closedness condition to: if $\nup(r)\subseteq\neff\U$,
then $\U^\ast\cap\head r=\emptyset$.
The antecedent needs then only be done once (since it only depends on $\U$),
whereas the consequent does not require consulting the database.

The following result summarizes these properties.
\begin{lemma}
  A weak repair $\U$ in a leaf of the justified repair tree for $\Ieta$ is a
  justified weak repair for $\Ieta$ iff, for every set $\U^\ast\subseteq\U$, 
  if $\nup(r)\subseteq\neff\U$, then $\U^\ast\cap\head r=\emptyset$.
\end{lemma}

The different implementations of repair trees use different subclasses of the
abstract class \verb+Node+; in particular, nodes of \verb+JustifiedRepairTree+s
must keep track not only of the sets of update actions being constructed, but
also of the sets of non-updatable actions that were assumed.
These labels are stored as \verb+Set<Action>+ using \verb+HashSet+ from the
Java library as implementation, as they are repeatedly tested
for membership everytime a new node is generated.

For efficiency, repair trees maintain internally a set of the sets of update
actions that label nodes constructed so far as a \verb+Set<Node>+.
This is used to avoid generating duplicate nodes with the same label.
Since this set is used mainly for querying, it is again implemented as a
\verb+HashSet+.
Nodes with inconsistent labels are also immediately eliminated, since they can
only produce inconsistent leaves.

\subsection{Interfacing with the User}

\begin{figure}[b]
  \centering
  \resizebox{.9\columnwidth}!{\includegraphics{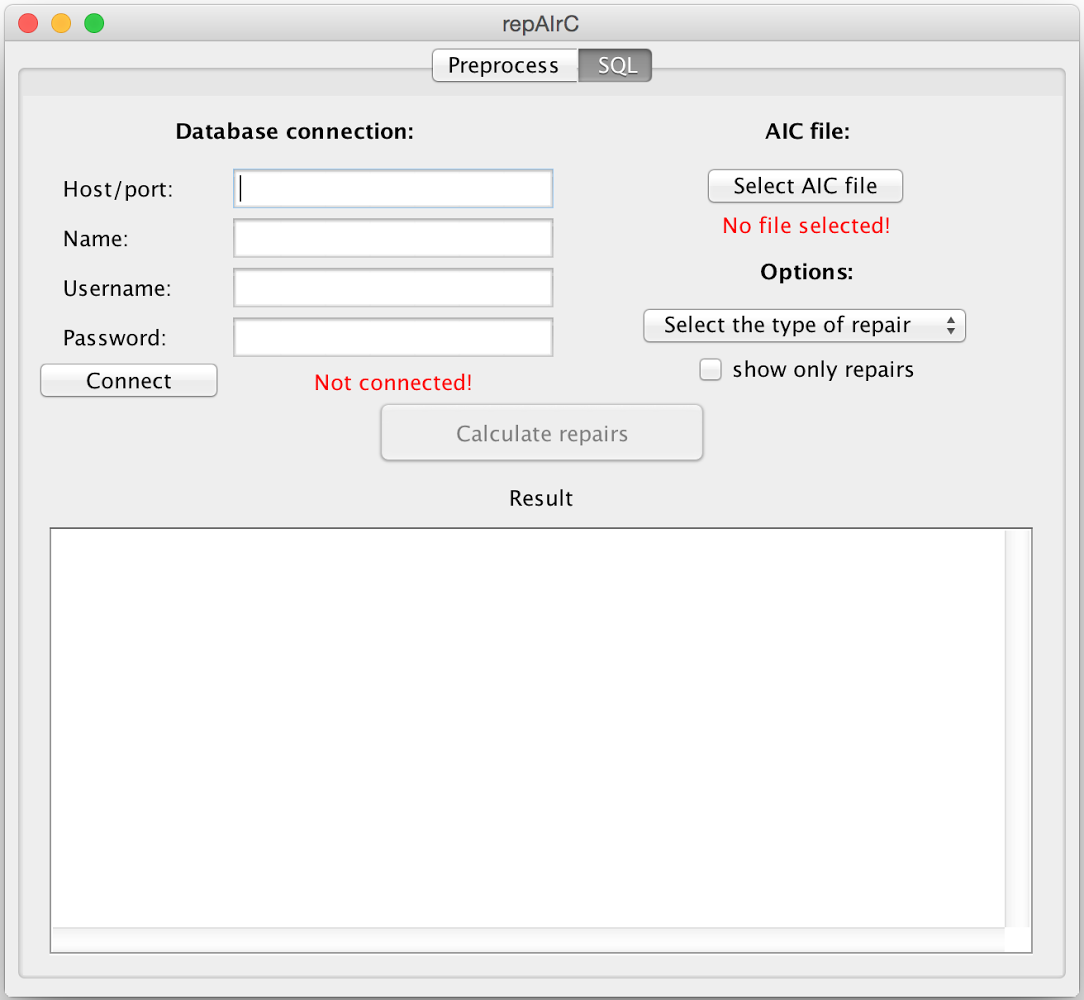}}
  \caption{The initial screen for \tool.}
  \label{fig:menu}
\end{figure}

\noindent
The user interface for \tool\ is implemented using the standard Java GUI widget
toolkit \verb+Swing+, and is rather straightforward.
On startup, the user is presented with the dialog box depicted in
Figure~\ref{fig:menu}.

The user can then provide credentials to connect to a database, as well as enter
a file containing a set of AICs.
If the connection to the database is successful and the file is successfully
parsed, \tool\ invokes the \verb+aicsCompatible+ method required by the
implementation of the \verb+DB+ interface (see Section~\ref{ssec:interface}) and
verifies that all tables and columns mentioned in the set of AICs are valid
tables and columns in the database.
If this is not the case, then an error message is generated and the user is
required to select new files; otherwise, the buttons for configuration and
computation of repairs become active.

Once the initialization has succeeded, one can check the database for
consistency and obtain different types of repairs, computed using the repair
tree described above.
As it may be of interest to obtain also weak repairs, the user is given the
possibility of selecting whether to
see only the repairs computed, or all valid leaves of the repair tree -- which
typically include some weak repairs.
In both cases the necessary validations are performed, so that leaves that do not
correspond to repairs (in the case of founded or justified repairs) are never
presented.

An example output screen after successful computation of the repairs for an
inconsistent database can be seen in Figure~\ref{fig:output}.

\begin{figure}[b]
  \centering
  \resizebox{.9\columnwidth}!{\includegraphics{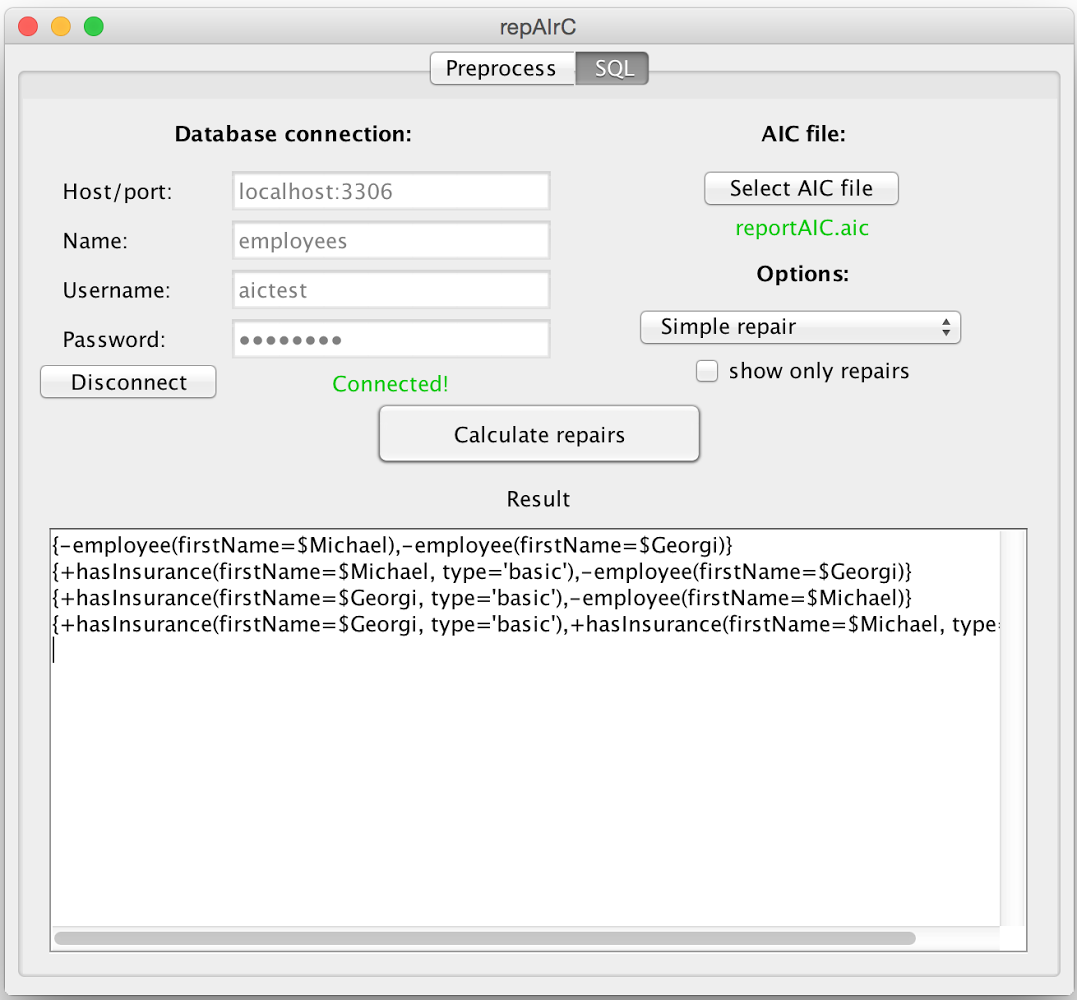}}
  \caption{Possible repairs of an inconsistent database.}
  \label{fig:output}
\end{figure}

\section{\uppercase{Parallelization and Stratification}}
\label{sec:parallel}

As described in Section~\ref{ssec:par}, it is possible to parallelize the search for
repairs of different kinds by splitting the set of AICs into independent sets;
in the case of founded or justified repairs, this parallelization can be taken
one step further by also stratifying the set of AICs.
Even though finding partitions and/or stratifications is asymptotically not
very expensive (it can be solved in linear time by the well-known graph
algorithms described below), it may still take noticeable time if the set of
AICs grows very large.

Since, by definition, partitions and stratifications are independent of the
actual database, it makes sense to avoid repeating their computation unless the
set of AICs changes.
For this reason, parallelization capabilities are implemented in \tool\ in a
two-stage process.
Inside \tool, the user can switch to the \verb+Preprocess+ tab, which provides
options for computing partitions and stratifications of a set of AICs.
This results in an annotated file which still can be read by the parser; in the
main tab, parallel computation is automatically enabled whenever the input file
is annotated in a proper manner.

\subsection{Implementation}

Computing optimal partitions in the spirit of~\cite{lcf:14} is not feasible in a
setting where variables are present, as this would require considering all
closed instances of all AICs -- but it is also not desirable, as it would also
result in a significant increase of the number of queries to the database.
Instead, we work with the adapted definition of dependency given in
Section~\ref{sec:background}.
Given a set of AICs, \tool\ constructs the adjacency matrix for the undirected
graph whose nodes are AICs and such that there is an edge between $r_1$ to $r_2$
iff $r_1$ and $r_2$ are not independent.
A partition is then computed simply by finding the connected components in this
graph by a standard graph algorithm.

The partitions computed are then written to a file, where each partition begins
with the line
\begin{verbatim}
#PARTITION_BEGIN_[NO]#
\end{verbatim}
where \verb+[NO]+ is the number of the current partition, and ends with
\begin{verbatim}
#PARTITION_END#
\end{verbatim}
and the AICs in each partition are inserted in between, in the standard format.

To compute the partitions for stratification, we need to find the strongly connected components of a
similar graph.
This is now a directed graph where there is an edge from $r_1$ to $r_2$ if $r_1$
precedes $r_2$.
The implementation is a variant of Tarjan's algorithm~\cite{Tarjan72}, adapted to
give also the dependencies between the connected components.

The computed stratification is then written to a file with a similar syntax to
the previous one, to which a dependency section is added, between the special
delimiters
\begin{verbatim}
#DEPENDENCIES_BEGIN#
\end{verbatim}
and
\begin{verbatim}
#DEPENDENCIES_END#
\end{verbatim}
\noindent The dependencies are included in this section as a sequence of strings
\verb+X -> Y+, one per line, where \verb+X+ and \verb+Y+ are the numbers of two
partitions
and \verb+Y+ precedes \verb+X+.

\begin{example}
\label{ex:strat}
  The two AICs from Example~\ref{ex:AIC1} cannot be parallelized, as they both
  use the \verb+junior+ table, but they can be stratified, as only the first one
  makes changes to this table.
  Preprocessing this example by \tool\ would return the following output.

{\small
\begin{verbatim}
#PARTITION_BEGIN_1#
junior(id = $X),
  category(type = boss, empId = $X)
  -> - junior(id = $X);
#PARTITION_END#
#PARTITION_BEGIN_2#
junior(id = $X),
  NOT insured(empId = $X, type = basic)
  -> + insured(empId = $X, type = basic);
#PARTITION_END#
#DEPENDENCIES_BEGIN#
2 -> 1
#DEPENDENCIES_END#
\end{verbatim}}

Imagine a simple scenario where the \verb+junior+ table contains a single entry.
Then, computing repairs for this set of AICs can be achieved by first repairing
partition $1$ (which will generate a tree with only one node) and then repairing
the resulting database w.r.t.~partition $2$ (which builds another tree, also
with only one node).
By comparison, processing the two AICs simultaneously would potentially give a
tree with $4$ nodes, as both AICs would have to be considered at each stage.
\end{example}

In general, if there are $n$ entries in the \verb+junior+ table, the stratified
approach will construct at most $n+1$ trees with a total of $n^2+n$ nodes (one
tree with $n$ nodes for the first AIC, at most $n$ trees with at most $n$ nodes
for the second AIC).
By contrast, processing both AICs together will construct a tree with
potentially $(2n)!$ leaves, which by removing duplicate nodes may still contain
$2^{2n}$ nodes.

This example shows that, by stratifying AICs, we can actually get an exponential
decrease on the size of the repair trees being built -- and therefore also on
the total runtime.

In addition to alleviating the exponential blowup of the repair trees, parallelization
and stratification also allow for a multi-threaded implementation, where repair trees
are built in parallel in multiple concurrent threads. To ensure that the dependencies between the
partitions are respected, the threads are instructed to wait for other threads that compute preceding
partitions. In Example~\ref{ex:strat}, the thread processing partition 2 would be instructed
to first wait for the thread processing partition 1 to finish.

Our empirical evaluation of \tool\ showed that speedups of a factor of $4$ to $7$
were observable even when processing small
parallelizable sets of only two or three AICs. For larger sets of AICs,
parallelization and stratification are necessary to obtain feasible runtimes. In one application,
which allowed for $15$ partitions to be processed independently, the stratified version computed
the founded repairs in approximately $1$ second, whereas the sequential version did not terminate
within a time limit of $15000$ seconds. This corresponds to a speedup of at least four orders of magnitude,
demonstrating the practical impact of the contributions of this section.

\subsection{Practical Assessment}

In the worst case, parallelization and stratification will have no impact on the
construction of the repair tree, as it is possible to construct a set of AICs
with no independent subsets.
However, the worst case is not the general case, and it is reasonable to believe
that real-life sets of AICs will actually have a high parallelization potential.

Indeed, integrity constraints typically reflect high-level consistency
requirements of the database, which in turn capture the hierarchical nature of
relational databases, where more complex relations are built from simpler ones.
Thus, when specifying \emph{active} integrity constraints there will naturally
be a preference to correct inconsistencies by updating the more complex tables
rather than the most primitive ones.

Furthermore, in a real setting we are not so much interested in repairing a
database once, but rather in ensuring that it remains consistent as its
information changes.
Therefore, it is likely that inconsistencies that arise will be localized to a
particular table.
The ability to process independent sets of AICs separately guarantees that we
will not be repeatedly evaluating those constraints that were not broken by
recent changes, focusing only on the constraints that can actually become
unsatisfied as we attempt to fix the inconsistency.

For the same reason, scalability of the techniques we implemented is not a relevant issue:
there is no practical need to develop a tool that is able to fix hundreds of inconsistencies
efficiently simultaneously, since each change to the database will likely only impact a few AICs.

\section{\uppercase{Conclusions and Future Work}}
\label{sec:concl}

\noindent
We presented a working prototype of a tool, called \tool, to check integrity of
real-world SQL databases with respect to a given set of
active integrity constraints, and to compute different
types of repairs automatically in case inconsistency is detected, following the ideas and algorithms
in~\cite{Flesca2004,Caroprese2007,Caroprese2011,CEGN2013,lcf:14}.
This tool is the first implementation of a concept we believe to have
the potential to be integrated in current database management systems.

Our tool currently does not automatically apply repairs to the database, rather presenting them to the user.
As discussed in \cite{Eiter1992}, such a functionality is not likely to be obtainable, as human intervention in the process of database repair is generally accepted to be necessary.
That said, automating the generation of a small and relevant set of repairs is a first important step in
ensuring a consistent data basis in Knowledge Management.

In order to deal with real-world heterogenous knowledge management systems, we
are currently working on extending and generalizing the notion of (active)
integrity constraints to encompass more complex knowledge repositories such as
ontologies, expert reasoning systems, and distributed knowledge bases.
The design of \tool\ has been with this extension in mind, and we believe that
its modularity will allow us to generalize it to work with such knowledge
management systems once the right theoretical framework is developed.

On the technical side, we are planning to speed up the system by integrating a local
database cache for peforming the many update and undo actions during exploration
of the repair trees without the overhead of an external database connection.

\section*{\uppercase{Acknowledgments}}

This work was supported by the Danish Council for Independent Research, Natural Sciences, and by FCT/MCTES/PIDDAC under centre grant to BioISI (Centre Reference: UID/MULTI/04046/2013).
Marta Ludovico was sponsored by a grant ``Bolsa Universidade de Lisboa / Funda\c c\~ao Amadeu Dias''.

\vfill
\bibliographystyle{apalike}
{\small
\bibliography{bibliografia}}

\begin{thebibliography}{}

\bibitem[Abiteboul, 1988]{Abiteboul1988}
Abiteboul, S. (1988).
\newblock Updates, a new frontier.
\newblock In Gyssens, M., Paredaens, J., and van Gucht, D., editors, {\em
  ICDT'88, 2nd International Conference on Database Theory, Bruges, Belgium,
  August 31 -- September 2, 1988, Proceedings}, volume 326 of {\em LNCS}, pages
  1--18. Springer.

\bibitem[Caroprese et~al., 2007]{Caroprese2007}
Caroprese, L., Greco, S., and Molinaro, C. (2007).
\newblock Prioritized active integrity constraints for database maintenance.
\newblock In Ramamohanarao, K., Krishna, P.~R., Mohania, M.~K., and
  Nantajeewarawat, E., editors, {\em Advances in Databases: Concepts, Systems
  and Applications, 12th International Conference on Database Systems for
  Advanced Applications, DASFAA 2007, Bangkok, Thailand, April 9-12, 2007,
  Proceedings}, volume 4443 of {\em LNCS}, pages 459--471. Springer.

\bibitem[Caroprese et~al., 2009]{Caroprese2009}
Caroprese, L., Greco, S., and Zumpano, E. (2009).
\newblock Active integrity constraints for database consistency maintenance.
\newblock {\em IEEE Transactions on Knowledge and Data Engineering},
  21(7):1042--1058.

\bibitem[Caroprese and Truszczy\'nski, 2011]{Caroprese2011}
Caroprese, L. and Truszczy\'nski, M. (2011).
\newblock Active integrity constraints and revision programming.
\newblock {\em Theory and Practice of Logic Programming}, 11(6):905--952.

\bibitem[Cruz-Filipe, 2014]{lcf:14}
Cruz-Filipe, L. (2014).
\newblock Optimizing computation of repairs from active integrity constraints.
\newblock In Beierle, C. and Meghini, C., editors, {\em Foundations of
  Information and Knowledge Systems - 8th International Symposium, FoIKS 2014,
  Bordeaux, France, March 3-7, 2014. Proceedings}, volume 8367 of {\em LNCS},
  pages 361--380. Springer.

\bibitem[Cruz-Filipe et~al., 2013]{CEGN2013}
Cruz-Filipe, L., Engr{\' a}cia, P., Gaspar, G., and Nunes, I. (2013).
\newblock Computing repairs from active integrity constraints.
\newblock In Wang, H. and Banach, R., editors, {\em 2013 International
  Symposium on Theoretical Aspects of Software Engineering, Birmingham, UK,
  July 1st--July 3rd 2013}, pages 183--190. IEEE.

\bibitem[Duhon, 1998]{Duhon98}
Duhon, B.~R. (1998).
\newblock It's all in our heads.
\newblock {\em Informatiktage}, 12(8):8--13.

\bibitem[Eiter and Gottlob, 1992]{Eiter1992}
Eiter, T. and Gottlob, G. (1992).
\newblock On the complexity of propositional knowledge base revision, updates,
  and counterfactuals.
\newblock {\em Artificial Intelligence}, 57(2--3):227--270.

\bibitem[Flesca et~al., 2004]{Flesca2004}
Flesca, S., Greco, S., and Zumpano, E. (2004).
\newblock Active integrity constraints.
\newblock In Moggi, E. and Scott~Warren, D., editors, {\em Proceedings of the
  6th International {ACM} {SIGPLAN} Conference on Principles and Practice of
  Declarative Programming, 24--26 August 2004, Verona, Italy}, pages 98--107.
  ACM.

\bibitem[Katsuno and Mendelzon, 1991]{Katsuno1991}
Katsuno, H. and Mendelzon, A.~O. (1991).
\newblock On the difference between updating a knowledge base and revising it.
\newblock In Allen, J.~F., Fikes, R., and Sandewall, E., editors, {\em
  Proceedings of the 2nd International Conference on Principles of Knowledge
  Representation and Reasoning (KR'91). Cambridge, MA, USA, April 22-25, 1991},
  pages 387--394. Morgan Kaufmann.

\bibitem[K{\"{o}}nig, 2012]{Konig2012}
K{\"{o}}nig, M.~E. (2012).
\newblock What is {KM}?
\newblock Knowledge Management Explained, \url{http://www.kmworld.com/}.

\bibitem[Tarjan, 1972]{Tarjan72}
Tarjan, R.~E. (1972).
\newblock Depth-first search and linear graph algorithms.
\newblock {\em {SIAM} Journal on Computing}, 1(2):146--160.

\bibitem[Winslett, 1990]{Winslett1990}
Winslett, M. (1990).
\newblock {\em Updating Logical Databases}.
\newblock Cambridge Tracts in Theoretical Computer Science. Cambridge
  University Press.

\end{thebibliography}

\end{document}